\documentclass[conference]{IEEEtran}

\usepackage{cite}
\usepackage{amsmath,amssymb,amsfonts}
\usepackage{amsthm}
\usepackage{bm}
\usepackage{bbm}
\usepackage{graphicx}
\usepackage{xcolor}
\usepackage{url}
\usepackage{enumitem}
\usepackage{subcaption}
\usepackage{balance}
\allowdisplaybreaks

\def\bv{\mathbf{v}}

\def\bs{\mathbf{s}}
\def\bx{\mathbf{x}}
\def\bu{\mathbf{u}}

\def\btheta{\boldsymbol{\theta}}
\def\bgamma{\boldsymbol{\gamma}}
\def\bzeta{\boldsymbol{\zeta}}

\def\bxi{\boldsymbol{\xi}}

\def\bI{\mathbf{I}}

\def\bC{\mathbf{C}}
\def\bG{\mathbf{G}}
\def\bSigma{\boldsymbol{\Sigma}}
\def\calM{\mathcal{M}}
\def\calQ{\mathcal{Q}}
\def\calD{\mathcal{D}}
\def\calS{\mathcal{S}}

\newtheorem{theorem}{Theorem}
\newtheorem{corollary}{Corollary}
\newtheorem{proposition}{Proposition}
\newtheorem{remark}{Remark}

\begin{document}

\title{Distributed Edge Learning under \\ Imperfect Data Sensing}

\author{
  \IEEEauthorblockN{Mehdi Karbalayghareh, David J.~Love, and Christopher G.~Brinton}
  \IEEEauthorblockA{
  Purdue University, West Lafayette, IN 47907, USA\\
  E-mails: \{mkarbala, djlove, cgb\}@purdue.edu}
}

\maketitle

\begin{abstract}
Distributed learning systems typically assume that local data is already available at clients with fixed quality, while in practice, data is sensed through \emph{imperfect} physical processes whose quality depends on modality, resolution, sensing power, and sample size. We model sensing noise as a structured modality-dependent covariance and derive a non-convex learning convergence bound whose irreducible sensing floor is governed by the alignment between the modality noise covariance and the loss-sensitivity geometry, so the optimal modality minimizes this noise-gradient alignment rather than total noise power alone. The analysis further yields a sensor-hardware achievability bound for $\epsilon$-stationarity and a hardware-saturation threshold on the accumulated dataset size. We jointly optimize modality, resolution, power, and sample count, and demonstrate the performance gain via simulations.
\end{abstract}


\section{Introduction}
\label{sec:intro}
AI-native wireless networks increasingly rely on distributed edge learning, where large numbers of heterogeneous devices must collaboratively train models under resource constraints~\cite{Andrews-6G,Giordani-6G-ComMag2020,Brinton-6G-ComMag2025,Saad-6G-2020}. Federated learning (FL) is an efficient framework for this regime: devices perform local gradient updates on private data and periodically share model parameters with a central server, avoiding the need to transmit raw data~\cite{McMahan2017FedAvg,Li2020FLSurvey,Kairouz2021AdvancesFL}.

Substantial work has analyzed FL convergence under practical impairments, including non-i.i.d.\ data, partial participation, imperfect channel conditions, and over-the-air aggregation distortion~\cite{Yang2020AirCompFL,Sery2021OTAHetero,Amiri2020ADSGD,Karbalayghareh-JSAIT26-dynamicFL}, but most analyses treat training data as a fixed \emph{noise-free} input. This abstraction breaks down when data is continuously sensed: unlike communication noise, which acts on model parameters post hoc and can be suppressed by power control or coding, sensing noise contaminates the training samples and propagates irreversibly into every gradient, accumulating into a persistent floor on the gradient norm governed by each device's sensing configuration.

In this work, we model the sensing process explicitly and characterize how sensor noise propagates through FL convergence. We introduce a structured per-modality noise covariance and derive a non-convex FL convergence bound whose irreducible floor is governed by the noise-gradient alignment between the sensor's directional noise and the loss-sensitivity geometry. We show that optimal modality selection is determined by this alignment rather than by total noise power alone, a distinction absent from prior sensing-aware FL formulations. The analysis further yields a hardware achievability condition on sensor parameters for $\epsilon$-stationarity, a hardware-saturation threshold on accumulated dataset size beyond which the device can no longer operate at maximum sensing power, and sensing configurations through a joint optimization for sensing power, modality, resolution, and sample count.

A brief outline of relevant literature is as follows. Existing work on sensing in FL has progressed along axes that do not model the noise of data acquisition via imperfect sensing. \emph{Modality availability heterogeneity} addresses clients possessing different modality subsets, with the design problem of model uploading or aggregation under heterogeneous configurations~\cite{mmFedMC,RELIEF}, but assumes clean modality inputs. \emph{Integrated sensing and communication (ISAC)} studies wireless channel distortion at ISAC front-ends~\cite{Cao2026VFEEL,Liu-STSP23_taskOriented,Wen-TWC2026-ISCC}, sample count optimization at the sensing side~\cite{Liang2025ISCC,Fu2025JSCC}, or radar-style detection of external targets via dual-purpose uplink signals~\cite{Asaad-Tabassum-2026,EM-PFL-ISAC-2025}. None of these directions models how the physical sensing parameters jointly shape the noise of the training samples themselves and affect FL performance. This is the gap that the present work addresses.

We assume full model delivery in the downlink and perfect uplink aggregation to \emph{isolate} the sensing effect. The joint design of sensing, communication, and computation under link impairments and imperfect over-the-air aggregation is treated in the journal version of this work.


\section{System Model}
\label{sec:system-model}
We consider an FL system with a parameter server and $N$ edge devices participating in the model training over $T$ communication rounds. Device~$n \in [N]$ begins with an initial local dataset $\calD_{n,0}$ of size $D_{n,0} = |\calD_{n,0}| \ge 0$. In each round $t \in \{0, 1, \ldots, T-1\}$, device~$n$ senses a fresh batch of $\overline{D}_{n,t}$ noisy samples $\overline{\calD}_{n,t}$ via the sensing process detailed in Section~\ref{sec:sensing}, and appends them to its accumulated dataset:
\begin{align}
  \calD_{n,t} = \calD_{n,t-1} \cup \overline{\calD}_{n,t}, \qquad D_{n,t} = D_{n,t-1} + \overline{D}_{n,t}. \label{eq:accumulation}
\end{align}

Let $F_{n,t}(\btheta) \triangleq \frac{1}{D_{n,t}} \sum_{\bu \in \calD_{n,t}} f(\btheta, \bu)$ denote the local loss at device~$n$ in round~$t$, where $f$ is the per-sample loss function. For a $d$-dimensional global model denoted by $\btheta \in \mathbb{R}^d$, the global loss function to be minimized is
\begin{align}
  F_t(\btheta) \triangleq \sum_{n=1}^{N} w_{n,t}\, F_{n,t}(\btheta), \label{eq:global-loss}
\end{align}
where $w_{n,t} \triangleq \frac{D_{n,t}}{D_t}$ with $D_t = \sum_n D_{n,t}$. Both the weights $w_{n,t}$ and the objective $F_t$ evolve deterministically with the accumulation history $\{\overline{D}_{n,\ell}\}_{\ell \leq t}$.

At the start of round~$t$, the server broadcasts the current global model $\btheta_t$. After acquiring $\overline{\calD}_{n,t}$ and forming $\calD_{n,t}$, device~$n$ performs $\tau$ steps of local mini-batch stochastic gradient descent (SGD) on $F_{n,t}$, initialized at $\btheta_{n,t}^{(1)} = \btheta_t$:
\begin{align}
  \btheta_{n,t}^{(i+1)} = \btheta_{n,t}^{(i)} - \eta\, \nabla F_{n}\!\bigl(\btheta_{n,t}^{(i)}; \bxi_{n,t}^{(i)}\bigr), \quad i \in [\tau], \label{eq:local-sgd}
\end{align}
where $\bxi_{n,t}^{(i)}$ is a mini-batch of fixed size $B$ drawn uniformly without replacement from $\calD_{n,t}$, and $\eta > 0$ is the stepsize. The per-device local update is $\Delta\btheta_{n,t} = \btheta_{n,t}^{(\tau+1)} - \btheta_t$, and the server forms the next global model as
\begin{align}
  \btheta_{t+1} = \btheta_t + \sum_{n=1}^{N} w_{n,t}\, \Delta\btheta_{n,t}. \label{eq:aggregation}
\end{align}

\subsection{Sensing Model}
\label{sec:sensing}
Let $\calM \triangleq \{1, \ldots, M_{\rm mod}\}$ be the set of available sensing modality types at the \emph{network level} (e.g., camera, microphone, LiDAR, RF sensor), and let $\calM_n \subseteq \calM$ denote the subset of modalities that device~$n$ is physically equipped to activate (its on-board sensor suite). A single-sensor device has $|\calM_n| = 1$, while a multi-sensor device can switch among its available modalities on a per-round basis. In each round~$t$, device~$n$ selects exactly one modality from its own suite via binary indicators $z_{n,t}^{(m)} \in \{0,1\}$ satisfying $z_{n,t}^{(m)} = 0$ for $m \notin \calM_n$ and $\sum_{m \in \calM_n} z_{n,t}^{(m)} = 1$, a sensing resolution $q_{n,t} \in \calQ_n = \{q_n^{(1)},\ldots,q_n^{(Q_n)}\}$ with maximum value $q_n^{\max} \triangleq \max\{\calQ_n\}$, and a sensing power level $p_{n,t}^s \in (0, P_{n}^{s,\max}]$, where $P_{n}^{s,\max}$ is the per-device sensing power constraint. 

Let $d_u$ denote the feature dimension. A sensed observation of an underlying clean signal $\bs_{n,t} \in \mathbb{R}^{d_u}$ is modeled as
\begin{align}
  \bx_{n,t} = \bs_{n,t} + \bv_{n,t}, \label{eq:sensing-model}
\end{align}
where $\bv_{n,t} \sim \mathcal{N}\big(\mathbf{0}, \bC_{n,t}^{(m_{n,t})}\big)$ is structured additive measurement noise with covariance
\begin{align}
  \bC_{n,t}^{(m)} = \frac{1}{p_{n,t}^s\, q_{n,t}}\, \bSigma_n^{(m)}, \label{eq:sensing-noise}
\end{align}
where $\bSigma_n^{(m)} \succeq 0$ is the modality-specific intrinsic noise covariance, encoding the directional noise structure of sensor~$m$ at device~$n$ in a shared feature space (after modality-specific preprocessing or feature extraction). Equation~\eqref{eq:sensing-noise} is an effective feature-space approximation rather than an exact raw-sensor law: the inverse scaling with sensing power reflects the improvement of effective measurement quality with acquisition energy, and the inverse scaling with resolution captures the reduction of quantization or reconstruction distortion at finer sensing fidelity. The separable form in~\eqref{eq:sensing-noise} provides a tractable model that preserves the key monotonic dependence of sensing quality on modality, power, and resolution. The isotropic special case $\bSigma_n^{(m)} = c_{1,n}^{(m)} \bI_{d_u}$ with $c_{1,n}^{(m)} > 0$ recovers the scalar noise variance $\bar{\sigma}_{n,t}^2 = c_{1,n}^{(m)}/(p_{n,t}^s q_{n,t})$, which is used in the simulations. The constant $c_{1,n}^{(m)}$ captures the intrinsic noise figure of sensor~$m$ at device~$n$ (e.g., sensor sensitivity, front-end noise floor).

The time to collect $\overline{D}_{n,t}$ samples under modality~$m$ is $T_{n,t}^s = \overline{D}_{n,t}/r_{{\rm s},n}^{(m)}$, where $r_{{\rm s},n}^{(m)}$ is the modality-specific sampling rate (samples per second). The corresponding sensing energy is
\begin{align}
  E_{n,t}^s = p_{n,t}^s\, T_{n,t}^s. \label{eq:sensing-energy}
\end{align}

\begin{remark}
The additive Gaussian noise model in \eqref{eq:sensing-noise} provides a tractable information-theoretic baseline for analyzing the tradeoffs among sensing energy, resolution, and information fidelity. Although some modalities exhibit structured or signal-dependent noise (e.g., Poisson noise in optical sensors or multiplicative noise in radar), the Gaussian assumption offers a unified framework and is worst-case among additive noises with fixed variance, so the resulting floors and achievability conditions provide upper bounds for practical design.
\end{remark}

\begin{remark}
The setup in~\eqref{eq:global-loss}--\eqref{eq:aggregation} requires each device to operate on a common model $\btheta$ despite physically distinct sensors. We therefore assume a pretrained modality-specific feature encoder $\phi^{(m)}: \mathbb{R}^{d_{\rm raw}^{(m)}} \to \mathbb{R}^{d_u}$ for each $m \in \calM$, mapping raw readings to a shared $d_u$-dimensional feature space (as in multimodal foundation models). Accordingly, $\bs_{n,t}$ is the clean feature-space signal, $\bv_{n,t}$ is raw-sensor noise propagated through $\phi^{(m)}$, and $\bSigma_n^{(m)}$ encodes both the intrinsic sensor noise and the encoder-induced geometry.
\end{remark}

\section{Convergence Analysis}
\label{sec:convergence}
We work under the following assumptions.
\begin{enumerate}[leftmargin=*,label=\emph{A\arabic*.}]

\item \emph{$L$-smoothness:} The loss $f(\btheta, \bu)$ is $L$-smooth in $\btheta$ for every $\bu$,
\begin{align}
  \|\nabla f(\btheta, \bu) - \nabla f(\btheta', \bu)\| \le L\|\btheta - \btheta'\|, \;\; \forall\, \btheta, \btheta', \bu. \label{eq:smoothness}
\end{align}
Thus, $F_{n,t}$ and $F_t$ are $L$-smooth in $\btheta$ for any $t$.

\item \emph{Bounded loss range:} There exist $F_{\inf}, F_{\sup} \in \mathbb{R}$ such that $F_{\inf} \le F_t(\btheta) \le F_{\sup}$ for all $\btheta$ and $t$. Let $\Delta_F \triangleq F_{\sup} - F_{\inf}$.

\item \emph{Structured gradient sensitivity:} For each device~$n$, there exists a positive semidefinite data-gradient sensitivity matrix $\bG_n(\btheta) \succeq 0$ such that, for any sample $\bu = \bs + \bv$ with sensing noise $\bv \sim \mathcal{N}(\mathbf{0}, \bC)$,
\begin{align}
  \mathbb{E}_{\bv}\bigl[\|\nabla f(\btheta;\bu) - \nabla f(\btheta;\bs)\|^2\bigr] \le \mathrm{tr}\bigl(\bG_n(\btheta)\,\bC\bigr). \label{eq:input-sensitivity}
\end{align}
This holds when $f(\btheta,\cdot)$ has a bounded cross-Jacobian $\nabla_{\btheta}\nabla_{\mathbf{u}}^\top f$. Under the isotropic case $\bC = \bar{\sigma}^2 \bI_{d_u}$, the bound in~\eqref{eq:input-sensitivity} reduces to $\mathrm{tr}(\bG_n(\btheta))\,\bar{\sigma}^2$.

\item \emph{First-order unbiased sensing perturbation:} Under the additive sensing model~\eqref{eq:sensing-model} with zero-mean noise, we adopt a first-order expansion of the sample gradient with respect to the sensed input. Under this approximation, the per-sample sensing perturbation is unbiased,
\begin{align}
  \mathbb{E}_{\bv}\bigl[\nabla f(\btheta;\bs+\bv) - \nabla f(\btheta;\bs)\bigr] = \mathbf{0}, \label{eq:first-order-unbiased}
\end{align}
which yields $\mathbb{E}[\bgamma_t^s \mid \btheta_t] = \mathbf{0}$ for the aggregated sensing perturbation in Section~\ref{sec:perturbation}.

\item \emph{Mini-batch and dataset-size variance decomposition:} Let
\[
F_n^\infty(\btheta) \triangleq \mathbb{E}_{\bu \sim \mathcal{P}_n}[f(\btheta,\bu)]
\]
denote the underlying device-level population loss. For the mini-batch $\bxi_{n,t}^{(i)}$ drawn uniformly from $\calD_{n,t}$, the stochastic gradient error with respect to the population gradient $\nabla F_n^\infty(\btheta)$ satisfies
\begin{align}
  \mathbb{E}\bigl[\|\nabla F_n(\btheta;\bxi_{n,t}^{(i)}) - \nabla F_n^\infty(\btheta)\|^2\bigr] \le \frac{\sigma_B^2}{B} + \frac{\sigma_{\rm pop}^2}{D_{n,t}}, \label{eq:sgd-variance}
\end{align}
for all $\btheta, n, t, i$. Here, $\sigma_B^2$ bounds the per-sample variance of the stochastic gradient around the empirical mean, and $\sigma_{\rm pop}^2$ bounds the per-sample variance around the population mean. The two need not be identical in general, and we keep them separate to make the dataset-size dependence visible.

\item \emph{Bounded sensitivity:} There exist finite constants \mbox{$\beta_n > 0$} such that the largest eigenvalue of $\bG_n(\btheta)$ satisfies $\lambda_{\max}(\bG_n(\btheta)) \le \beta_n$ for all $\btheta$ and $n$.

\item \emph{Bounded gradient norm:} There exists $G > 0$ such that $\mathbb{E}\|\nabla F_n(\btheta; \calS)\|^2 \le G^2$ for any nonempty $\calS \subseteq \calD_{n,t}$, all $\btheta, n, t$. This bound is standard in non-convex FL convergence analysis~\cite{Liang2025ISCC,Wen-TWC2026-ISCC}.

\item \emph{Bounded gradient heterogeneity:} There exists $\Gamma \ge 0$ such that
\begin{align}
  \mathbb{E}\bigl[\|\nabla F_{n,t}(\btheta) - \nabla F_t(\btheta)\|^2\bigr] \le \Gamma^2, \;\; \forall\, \btheta, n, t. \label{eq:heterogeneity}
\end{align}
This captures the heterogeneity of local empirical losses across devices and is standard in analyses with $\tau \ge 1$ local steps~\cite{Li2020FedProx,Karimireddy-SCAFFOLD}.
\end{enumerate}

\subsection{Perturbation Decomposition}
\label{sec:perturbation}
We write the global update rule as a perturbed full-gradient counterpart of the $\tau$ local SGD steps on the objective $F_t$:
\begin{align}
  \btheta_{t+1} = \btheta_t - \eta\tau \nabla F_t(\btheta_t) + \bzeta_t, \label{eq:perturbed-update}
\end{align}
where $\bzeta_t$ collects all deviations of the actual aggregation update from the ideal $\tau$-step full-gradient reference on $F_t$. Three independent sources of perturbation contribute: sensing noise embedded in the accumulated data, denoted $\bgamma_t^s$; mini-batch variance and client drift from $\tau$ steps of local SGD, denoted $\bgamma_t^c$; and dataset drift caused by appending $\overline{\calD}_{n,t+1}$ to each device's accumulated set, denoted $\bgamma_t^d$. Thus, $\bzeta_t = \bgamma_t^s + \bgamma_t^c + \bgamma_t^d$.

\subsubsection{Sensing noise contribution}
Each sample $\bu \in \calD_{n,t}$ was acquired in some round $\ell \le t$ under sensing configuration $(m_{n,\ell}, p_{n,\ell}^s, q_{n,\ell})$, with sensing noise realizations independent across samples. Defining the noise-gradient alignment coefficient
\begin{align}
  \kappa_n^{(m)}(\btheta) \triangleq \mathrm{tr}\bigl(\bG_n(\btheta)\,\bSigma_n^{(m)}\bigr), \label{eq:kappa-def}
\end{align}
and applying A3 sample-by-sample,
\begin{align}
  \mathbb{E}\bigl[\|\bgamma_{n, t}^{s}\|^2\bigr] \le \frac{1}{D_{n,t}^2} \sum_{\ell=0}^{t} \frac{\overline{D}_{n,\ell}\, \kappa_n^{(m_{n,\ell})}(\btheta_t)}{p_{n,\ell}^s\, q_{n,\ell}}, \label{eq:sensing-per-device}
\end{align}
where $\bgamma_{n, t}^{s}$ is the sensing-induced gradient perturbation at device~$n$. Aggregating across devices with weights $w_{n,t}$,
\begin{align}
  \mathbb{E}\|\bgamma_t^s\|^2 \le \sum_{n=1}^{N} \frac{w_{n,t}}{D_{n,t}^2} \sum_{\ell=0}^{t} \frac{\overline{D}_{n,\ell}\, \kappa_n^{(m_{n,\ell})}(\btheta_t)}{p_{n,\ell}^s\, q_{n,\ell}} \triangleq \mathcal{E}_t^s. \label{eq:sensing-bound}
\end{align}
Under A6, $\kappa_n^{(m)}(\btheta) \le \tilde{\kappa}_n^{(m)} \triangleq \beta_n\,\mathrm{tr}(\bSigma_n^{(m)})$ uniformly in $\btheta$. 


\subsubsection{Mini-batch and local-step contribution}
With $\tau$ local mini-batch SGD steps per round, the variance decomposition in A5, and gradient heterogeneity $\Gamma$ in A8, the local-SGD perturbation analysis~\cite{Li2020FedProx,Karimireddy-SCAFFOLD} gives
\begin{align}
  \mathbb{E}\|\bgamma_t^c\|^2 \le \eta^2\tau \!\sum_{n=1}^{N}\! w_{n,t}\!\left(\!\frac{\sigma_B^2}{B} + \frac{\sigma_{\rm pop}^2}{D_{n,t}}\!\right) + \eta^2\tau^2\Gamma^2 \triangleq \mathcal{E}_t^c. \label{eq:comp-bound}
\end{align}
The first term denotes the fixed mini-batch variance, while the second reflects the empirical-to-population discrepancy and decreases as the accumulated dataset grows. The last term is the client-drift contribution arising from local descent on heterogeneous $F_{n,t}$ rather than the global $F_t$. 

\subsubsection{Dataset drift contribution}
Appending $\overline{\calD}_{n,t+1}$ to each device's accumulated dataset changes the global objective from $F_t$ to $F_{t+1}$. Following the decomposition technique of~\cite[Lemma~1]{Liang2025ISCC}, the gradient of $F_{t+1}$ at any $\btheta$ admits the convex combination
\begin{align}
  \nabla &F_{t+1}(\btheta) \nonumber\\
  &= \!\sum_{n=1}^{N}\,\Bigl[ \tfrac{D_{n,t}}{D_{t+1}} \nabla F_{n,t}(\btheta) + \tfrac{\overline{D}_{n,t+1}}{D_{t+1}} \nabla F_n(\btheta; \overline{\calD}_{n,t+1}) \Bigr]. \label{eq:lemma1-decomp}
\end{align}
The model update $\btheta_t \to \btheta_{t+1}$ is driven by the round-$t$ gradients (over $\calD_{n,t}$), but the round-$(t+1)$ objective $F_{t+1}$ is partly determined by the freshly sensed batches $\overline{\calD}_{n,t+1}$. The residual $\bgamma_t^d \triangleq \nabla F_{t+1}(\btheta_t) - \nabla F_t(\btheta_t)$ captures this drift, and is bounded under A7 by
\begin{align}
  \mathbb{E}\|\bgamma_t^d\|^2 \le \!\sum_{n=1}^{N} \frac{2\overline{D}_{n,t+1}}{D_{t+1}}\, G^2 \triangleq \mathcal{E}_t^d. \label{eq:drift-bound}
\end{align}

\subsubsection{Total perturbation power}
By A4, $\bgamma_t^s$ is conditionally zero-mean, while $\bgamma_t^c$ may have a nonzero conditional mean due to local-model drift and client heterogeneity, and $\bgamma_t^d$ is deterministic given $\btheta_t$. Thus the cross terms involving $\bgamma_t^s$ vanish under expectation, giving $\mathbb{E}\|\bzeta_t\|^2 = \mathbb{E}\|\bgamma_t^s\|^2 + \mathbb{E}\|\bgamma_t^c + \bgamma_t^d\|^2$. Applying $\|a+b\|^2 \le 2\|a\|^2 + 2\|b\|^2$ to the second term,
\begin{align}
  \mathbb{E}\|\bzeta_t\|^2 \le \mathcal{E}_t^s + 2(\mathcal{E}_t^c + \mathcal{E}_t^d) \triangleq \Phi_t. \label{eq:total-perturbation}
\end{align}
All three components $\mathcal{E}_t^s$, $\mathcal{E}_t^c$, and $\mathcal{E}_t^d$ are now controllable through the sensing decisions $\{(m_{n,\ell}, p_{n,\ell}^s, q_{n,\ell}, \overline{D}_{n,\ell})\}_{\ell \le t}$ via the accumulated dataset sizes $D_{n,\ell}$, and they exhibit a nontrivial joint tradeoff in $\overline{D}_{n,\ell}$ formalized in Section~\ref{sec:implications}.

\subsection{Main Convergence Result}

\begin{theorem}
\label{thm:convergence}
Under assumptions A1--A8, let the global model evolve according to~\eqref{eq:aggregation}, and let $\Phi_t$ be the per-round perturbation power from~\eqref{eq:total-perturbation}. Define $\delta_t \triangleq \Delta_F \sum_{n=1}^{N} \overline{D}_{n,t+1}/D_{t+1}$, where $\Delta_F = F_{\sup} - F_{\inf}$ from A2. If the stepsize satisfies $0 < \eta\tau \le 1/L$, then
\vspace{-0.05in}
\begin{align}
  \frac{1}{T}&\sum_{t=0}^{T-1}  \mathbb{E}\bigl\|\nabla F_t(\btheta_t)\bigr\|^2 \le \frac{4\bigl(\mathbb{E}[F_0(\btheta_0)] - F_{\inf}\bigr)}{\eta\tau T} \notag \\
  & + \frac{2L}{\eta\tau T}\sum_{t=0}^{T-1}\Phi_t + \frac{8}{\eta^2\tau^2 T}\sum_{t=0}^{T-1}\bigl(\mathcal{E}_t^c + \mathcal{E}_t^d\bigr) + \frac{4}{\eta\tau T}\sum_{t=0}^{T-1}\delta_t. \label{eq:main-bound}
\end{align}
\end{theorem}

\begin{proof}
By A1, $F_t$ is $L$-smooth in $\btheta$. Applying~\eqref{eq:smoothness} with $\btheta' = \btheta_{t+1}$ and substituting $\btheta_{t+1} - \btheta_t = -\eta\tau\nabla F_t(\btheta_t) + \bzeta_t$,
\begin{align}
  F_t(\btheta_{t+1}) &\le F_t(\btheta_t) - \eta\tau\bigl(1 - \tfrac{L\eta\tau}{2}\bigr)\|\nabla F_t(\btheta_t)\|^2 \notag \\ 
  &\quad + (1 - L\eta\tau)\,\nabla F_t(\btheta_t)^\top \bzeta_t + \tfrac{L}{2}\|\bzeta_t\|^2.
\end{align}
Conditioning on $\btheta_t$, the sensing perturbation satisfies $\mathbb{E}[\bgamma_t^s \mid \btheta_t] = \mathbf{0}$ by~\eqref{eq:first-order-unbiased}, while $\bar{\bgamma}_t \triangleq \mathbb{E}[\bgamma_t^c \mid \btheta_t] + \bgamma_t^d$ is generally nonzero (client drift plus deterministic dataset shift). For $\eta\tau \le 1/L$, $0 \le 1 - L\eta\tau \le 1$, and applying Cauchy--Schwarz with AM--GM weight $\eta\tau/2$,
\begin{align}
  (1 - L\eta\tau)\nabla F_t(\btheta_t)^\top \bar{\bgamma}_t \le \tfrac{\eta\tau}{4}\|\nabla F_t(\btheta_t)\|^2 + \tfrac{1}{\eta\tau}\|\bar{\bgamma}_t\|^2.
\end{align}
Since $\eta\tau(1 - L\eta\tau/2) \ge \eta\tau/2$ for $\eta\tau \le 1/L$, the net coefficient of $\|\nabla F_t(\btheta_t)\|^2$ is at most $-\eta\tau/4$. Taking total expectation, using Jensen's inequality $\mathbb{E}\|\mathbb{E}[\bgamma_t^c \mid \btheta_t]\|^2 \le \mathbb{E}\|\bgamma_t^c\|^2 \le \mathcal{E}_t^c$, and $\|\bar{\bgamma}_t\|^2 \le 2\|\mathbb{E}[\bgamma_t^c\mid\btheta_t]\|^2 + 2\|\bgamma_t^d\|^2$,
\begin{align}
  \mathbb{E}[F_t(\btheta_{t+1})] \le \mathbb{E}[F_t(\btheta_t)] &- \tfrac{\eta\tau}{4}\mathbb{E}\|\nabla F_t(\btheta_t)\|^2 \notag \\ 
  & + \tfrac{2(\mathcal{E}_t^c + \mathcal{E}_t^d)}{\eta\tau} + \tfrac{L\Phi_t}{2}. \label{eq:descent-final}
\end{align}
By A2 and~\eqref{eq:lemma1-decomp}, $|F_{t+1}(\btheta) - F_t(\btheta)| \le \delta_t$ for all $\btheta$, so $\mathbb{E}[F_{t+1}(\btheta_{t+1})] \le \mathbb{E}[F_t(\btheta_{t+1})] + \delta_t$. Combining with~\eqref{eq:descent-final} and telescoping over $t = 0, \ldots, T-1$, the $\mathbb{E}[F_t(\btheta_t)]$ terms cancel and we obtain
\begin{align}
  \mathbb{E}[F_T(\btheta_T)] &- \mathbb{E}[F_0(\btheta_0)] \le -\tfrac{\eta\tau}{4}\sum_{t=0}^{T-1}\mathbb{E}\|\nabla F_t(\btheta_t)\|^2 \notag \\
  & + \tfrac{2}{\eta\tau}\sum_{t=0}^{T-1}(\mathcal{E}_t^c + \mathcal{E}_t^d) + \tfrac{L}{2}\sum_{t=0}^{T-1}\Phi_t + \sum_{t=0}^{T-1}\delta_t.
\end{align}
Using $\mathbb{E}[F_T(\btheta_T)] \ge F_{\inf}$ from A2, rearranging, and dividing by $\eta\tau T/4$ yields~\eqref{eq:main-bound}.
\end{proof}


\subsection{Implications of the Sensing Floor}
\label{sec:implications}

The sensing floor $\mathcal{E}_t^s$ in~\eqref{eq:sensing-bound} couples the noise covariance $\bSigma_n^{(m)}$ to the loss-sensitivity geometry $\bG_n(\btheta_t)$ through the alignment coefficient $\kappa_n^{(m)}(\btheta_t)$, and aggregates contributions from all rounds $\ell \le t$ in which samples were acquired. 

\begin{proposition}
\label{prop:alignment}
The round-$t$ modality $m_{n,t}$ contributes $\frac{w_{n,t}\, \overline{D}_{n,t}\, \kappa_n^{(m_{n,t})}(\btheta_t)}{D_{n,t}^2\, p_{n,t}^s\, q_{n,t}}$ to $\mathcal{E}_t^s$. For modalities $m_1, m_2 \in \calM_n$ with equal total noise power $\mathrm{tr}(\bSigma_n^{(m_1)}) = \mathrm{tr}(\bSigma_n^{(m_2)})$, $m_1$ strictly dominates $m_2$ if and only if $\mathrm{tr}(\bG_n(\btheta_t)\bSigma_n^{(m_1)}) < \mathrm{tr}(\bG_n(\btheta_t)\bSigma_n^{(m_2)})$. Hence, modality selection should be based on the noise-gradient alignment $\kappa_n^{(m)}(\theta_t)$, rather than on total noise power alone.
\end{proposition}

\begin{proposition}[Sample count tradeoff]
\label{prop:sample-tradeoff}
The per-round sample count $\overline{D}_{n,t}$ enters the convergence bound~\eqref{eq:main-bound} through two opposing channels:
\begin{enumerate}
\item \emph{Statistical benefit:} Larger $\overline{D}_{n,t}$ enlarges $D_{n,t}$, which (i) decreases the dataset-size-dependent term in $\mathcal{E}_t^c$ through the $\sigma_{\rm pop}^2/D_{n,t}$ component of A5, and (ii) increases $D_{n,t}^2$ in the sensing-floor denominator, diluting the per-sample noise contribution to $\mathcal{E}_t^s$ across the accumulated dataset.
\item \emph{Acquisition and drift cost:} Larger $\overline{D}_{n,t}$ increases the sensing energy $E_{n,t}^s = p_{n,t}^s\, \overline{D}_{n,t}/r_{{\rm s},n}^{(m_{n,t})}$, lengthens the sensing time $T_{n,t}^s$, and increases the round-$(t-1)$ drift bridge $\delta_{t-1} \propto \sum_n \overline{D}_{n,t}/D_t$.
\end{enumerate}
The optimal $\overline{D}_{n,t}$ balances these competing effects, captured by the joint optimization in Section~\ref{sec:optimization}.
\end{proposition}

\begin{corollary}
\label{cor:floor}
Under the conditions of Theorem~\ref{thm:convergence}, for any feasible sensing configuration with $0 < p_{n,\ell}^s \le P_n^{s,\max}$ and $q_{n,\ell} \le q_n^{\max}$ for $\ell = 0, \ldots, t$,
\begin{align}
  \mathcal{E}_t^s \ge \!\sum_{n=1}^{N}\!\frac{w_{n,t}}{D_{n,t}^2 P_n^{s,\max} q_n^{\max}}\!\sum_{\ell=0}^{t}\!\overline{D}_{n,\ell}\, \kappa_n^{(m_{n,\ell})}(\btheta_t) > 0,\!\!\! \label{eq:floor-lower-bound}
\end{align}
whenever $\kappa_n^{(m)}(\btheta_t) > 0$ for any active modality at any device, i.e., whenever the sensor covariance shares a nonzero eigenvector with the loss-sensitivity matrix. Under bounded total accumulation, the floor term in~\eqref{eq:main-bound} is governed by sensor hardware constraints and cannot be removed by algorithmic choices alone.
\end{corollary}

\begin{proof}
Substituting the maximum feasible values $p_{n,\ell}^s = P_n^{s,\max}$ and $q_{n,\ell} = q_n^{\max}$ for all $\ell \le t$ into~\eqref{eq:sensing-bound} yields the lower bound. The right-hand side (RHS) is strictly positive by the assumption on $\kappa_n^{(m)}$.
\end{proof}

Define the time-averaged minimum sensing floor (using A6 and the best modality available at each device) as
\begin{align}
  \bar{\mathcal{E}}^{s,\min} \triangleq \frac{1}{T}\sum_{t=0}^{T-1} \!\sum_{n=1}^{N}\!\frac{w_{n,t}\, \min_{m \in \calM_n} \tilde{\kappa}_n^{(m)}}{D_{n,t}^2\, P_n^{s,\max}\, q_n^{\max}}\!\sum_{\ell=0}^{t}\!\overline{D}_{n,\ell}. \label{eq:Esmin}
\end{align}

\begin{theorem}
\label{thm:achievability}
Under the conditions of Theorem~\ref{thm:convergence}, define the time-averaged controllable floor $\bar{\mathcal{E}}^c \triangleq \frac{1}{T}\sum_t (\mathcal{E}_t^c + \mathcal{E}_t^d)$ and the time-averaged drift $\bar{\delta} \triangleq \frac{1}{T}\sum_t \delta_t$. A necessary condition for $\epsilon$-stationarity to be attainable for any $T$ is
\begin{align}
  2L\,\bar{\mathcal{E}}^{s,\min} + \bigl(4L + 8/(\eta\tau)\bigr)\bar{\mathcal{E}}^c + 4\bar{\delta} < \eta\tau\,\epsilon. \label{eq:achievability-cond}
\end{align}
If~\eqref{eq:achievability-cond} fails, the convergence bound~\eqref{eq:main-bound} cannot guarantee $\frac{1}{T}\sum_{t}\mathbb{E}\|\nabla F_t(\btheta_t)\|^2 < \epsilon$ regardless of training duration or sensing parameters.
\end{theorem}
\begin{proof}
The RHS of~\eqref{eq:main-bound} is a sum of nonnegative terms. Time-averaging and using~\eqref{eq:total-perturbation} together with the lower bound~\eqref{eq:Esmin} yields a lower bound on the RHS of~\eqref{eq:main-bound}:
\[
\text{RHS of~\eqref{eq:main-bound}}
\ge
\frac{2L}{\eta\tau}\bar{\mathcal{E}}^{s,\min}
+
\frac{4L+8/(\eta\tau)}{\eta\tau}\bar{\mathcal{E}}^c
+
\frac{4}{\eta\tau}\bar{\delta},
\]
up to the vanishing initialization term. For~\eqref{eq:main-bound} to certify $\epsilon$-stationarity, this lower bound on its RHS must itself be below $\epsilon$, giving the stated necessary condition.
\end{proof}


\section{Optimal Sensing Configuration}
\label{sec:optimization}
The convergence bound~\eqref{eq:main-bound} shows that all three controllable floor components $\bar{\mathcal{E}}^s$, $\bar{\mathcal{E}}^c_t$, and $\bar{\mathcal{E}}^d_t$ depend on the per-round sensing decisions through the accumulated dataset sizes $D_{n,t}$. According to Proposition~\ref{prop:sample-tradeoff}, the sample count $\overline{D}_{n,t}$ is itself a per-round decision variable. For each device~$n$ in round $t$, we seek a sensing configuration $(m_{n,t}, q_{n,t}, p_{n,t}^s, \overline{D}_{n,t})$ that jointly controls convergence quality and sensing energy. Adopting a weighted-sum formulation with nonnegative weights $(\omega_L, \omega_C, \omega_E)$ on the sensing floor, mini-batch variance, and sensing energy contributions respectively, we solve
\begin{align}
\mathcal{P}: \min_{\substack{m_{n,t}, q_{n,t}, \\ p_{n,t}^s, \overline{D}_{n,t}}} \;& \omega_L\, \frac{\tilde{\kappa}_n^{(m_{n,t})}\,\overline{D}_{n,t}}{D_{n,t}^2\, p_{n,t}^s\, q_{n,t}} + \omega_C\, \frac{\eta^2\tau\sigma_{\rm pop}^2}{D_{n,t}} \notag\\ 
& \quad + \omega_E\, \frac{p_{n,t}^s\, \overline{D}_{n,t}}{r_{{\rm s},n}^{(m_{n,t})}} \label{eq:joint-opt}\\
\mathrm{s.t.}\;  m_{n,t} &\,\in \calM_n,\;\, q_{n,t} \in \calQ_n,\;\, \overline{D}_{n,t} \in \{1,\ldots,\overline{D}_n^{\max}\}, \notag \\
& 0 < p_{n,t}^s \le P_n^{s,\max}, \notag
\end{align}
where $D_{n,t} = D_{n,t-1}+\overline{D}_{n,t}$, $\tilde{\kappa}_n^{(m)} = \beta_n\,\mathrm{tr}(\bSigma_n^{(m)})$ via A6 bound, and $\overline{D}_n^{\max}$ is the maximum per-round sensing samples. Problem $\mathcal{P}$ is solved in two steps: closed-form $p_{n,t}^s$ for fixed $(m_{n,t}, q_{n,t}, \overline{D}_{n,t})$, then finite search over $(m_{n,t}, q_{n,t}, \overline{D}_{n,t})$.
 
\subsection{Optimal Sensing Power}
\label{sec:pstar}
 
For fixed $(m_{n,t}, q_{n,t}, \overline{D}_{n,t})$, only the first and third terms of $\mathcal{P}$ depend on $p_{n,t}^s$, and their sum is strictly convex. Setting the derivative to zero gives
\begin{align}
  p_{n,t}^{s\star} = \min\!\left(\!P_n^{s,\max},\;\sqrt{\frac{\omega_L\,\tilde{\kappa}_n^{(m_{n,t})}\, r_{{\rm s},n}^{(m_{n,t})}}{\omega_E\, D_{n,t}^2\, q_{n,t}}}\right). \label{eq:pstar}
\end{align}

\begin{corollary}[Hardware-saturation threshold]
\label{cor:sample-threshold}
For fixed $(m, q)$, the minimizer $p_{n,t}^{s\star}$ in \eqref{eq:pstar} attains the hardware cap $P_n^{s,\max}$ when the accumulated dataset reaches the threshold
\begin{align}
  D_n^\dagger(m,q) \triangleq \frac{1}{P_n^{s,\max}}\sqrt{\frac{\omega_L\,\tilde{\kappa}_n^{(m)}\, r_{{\rm s},n}^{(m)}}{\omega_E\, q}}. 
  \label{eq:Dstar}
\end{align}
For $D_{n,t} \le D_n^\dagger$, the device operates at maximum sensing power $P_n^{s,\max}$, and the per-round sensing floor contribution scales as $\overline{D}_{n,t}/(D_{n,t}^2\, P_n^{s,\max}\, q)$. For $D_{n,t} > D_n^\dagger$, the unconstrained optimum binds, $p_{n,t}^{s\star} \propto 1/D_{n,t}$, and substitution into the floor yields a contribution scaling as $\overline{D}_{n,t}/D_{n,t}$. In both regimes, the round-$t$ floor contribution is bounded above and decreases with the accumulated dataset, reflecting the dominance of the historical-averaging effect over the per-sample noise increase from reduced power.
\end{corollary}
 
\subsection{Joint Search Over Modality, Resolution, and Sample Count}
\label{sec:discrete}
 
With $p_{n,t}^{s\star}$ from~\eqref{eq:pstar}, the discrete values $(m_{n,t}, q_{n,t}, \overline{D}_{n,t})$ are obtained by minimizing the combined sensing cost 
\begin{align*}
  J_{n,t}^{\rm disc} \triangleq \omega_L\, \frac{\tilde{\kappa}_n^{(m_{n,t})}\,\overline{D}_{n,t}}{D_{n,t}^2\, p_{n,t}^{s\star}\, q_{n,t}} + \omega_C\, \frac{\eta^2\tau\sigma_{\rm pop}^2}{D_{n,t}} +\; \omega_E\, \frac{p_{n,t}^{s\star}\,\overline{D}_{n,t}}{r_{{\rm s},n}^{(m_{n,t})}} \label{eq:disc-cost}
\end{align*}
over the finite set $\calM_n \times \calQ_n \times \{1,\ldots,\overline{D}_n^{\max}\}$. The objective in $\overline{D}_{n,t}$ is unimodal: the mini-batch term decreases while the floor and energy terms increase, yielding a unique interior minimizer for each $(m_{n,t}, q_{n,t})$. Proposition~\ref{prop:alignment} identifies $\kappa_n^{(m)}(\btheta_t)$ as the exact modality criterion; $\tilde{\kappa}_n^{(m)}$ is used as a tractable upper bound when $\bG_n(\btheta_t)$ is unavailable. The search has complexity $\mathcal{O}(|\calM_n|\,|\calQ_n|\,\overline{D}_n^{\max})$ per device per round.

\section{Numerical Results}
\label{sec:simulations}
Experiments use the MNIST dataset with a CNN model. Here, $N = 10$, $\tau = 5$, $B=16$, and $T = 100$. Both i.i.d.\ and non-i.i.d.\ data partitions are evaluated across devices. We assume that three modalities are available with $(c_1^{(m)}, r_s^{(m)})$ values $(1.0, 20)$, $(0.5, 10)$, and $(0.2, 1)$ for $m=1,2,3$ respectively, under the isotropic model $\bSigma_n^{(m)} = c_1^{(m)}\bI_{d_u}$, so $\tilde{\kappa}_n^{(m)} \propto c_1^{(m)}$. Thus, $m=1$ and $m=3$ are the worst and the best modalities, respectively. The 10 devices are split into three groups: 4 with $\calM_n=\{1\}$, 3 with $\calM_n=\{1,2\}$, and 3 with $\calM_n=\{1,2,3\}$. The resolution set is $\calQ_n=\{1,2,4\}$ for all devices, $P_n^{s,\max}=1$~W, and $\overline{D}_n^{\max}=50$ per round.
 
 
 
 
 
\begin{figure}[!t]
  \centering
  \begin{subfigure}[t]{0.45\linewidth}
    \centering
    \includegraphics[width=\linewidth]{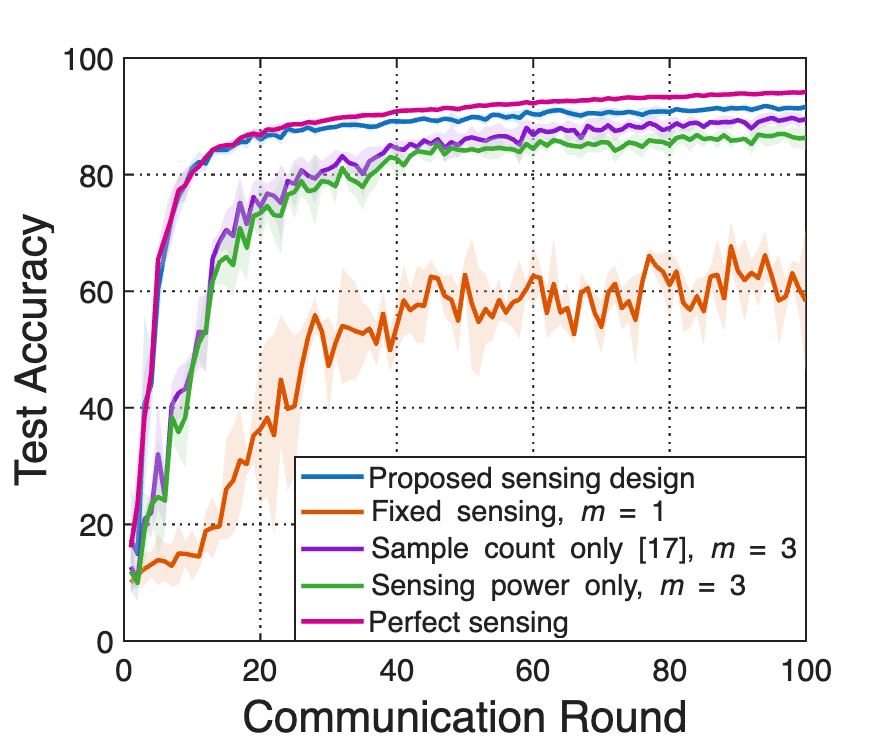}
    \caption{i.i.d.}
    \label{fig:acc_iid}
  \end{subfigure}
  \hfill
  \begin{subfigure}[t]{0.45\linewidth}
    \centering
    \includegraphics[width=\linewidth]{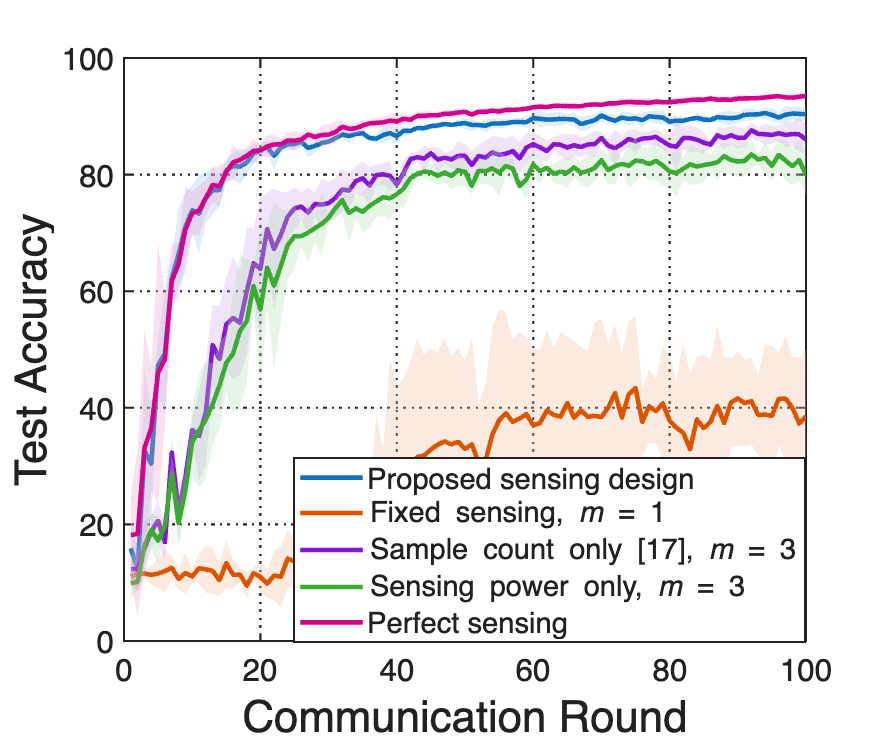}
    \caption{non-i.i.d.}
    \label{fig:acc_noniid}
  \end{subfigure}
  \caption{Test accuracy vs.\ communication rounds (MNIST).\vspace{-0.2in}}
  \label{fig:accuracy}
\end{figure}
 
Fig.~\ref{fig:accuracy} compares the test accuracy across several benchmarks. The proposed sensing-aware design outperforms all baselines and remains very close to the perfect sensing oracle, where the sensed data are noiseless. The sample count control baseline, following the sensing-side sample-size optimization in~\cite{Liang2025ISCC}, optimizes only the number of newly sensed samples while we fix its modality to the best one, $m=3$; it still underperforms the proposed scheme. Results confirm the benefit of jointly adapting all sensing variables.

\section{Conclusion}
\label{sec:conclusion}
This paper studies imperfect multimodal sensing in distributed edge learning, deriving a non-convex convergence bound under structured per-modality sensing noise and data accumulation. The irreducible sensing floor is governed by the noise-gradient alignment $\kappa_n^{(m)}$, so optimal modality selection minimizes this alignment rather than total noise power alone. The analysis yields a hardware achievability condition for $\epsilon$-stationarity, a saturation threshold $D_n^\dagger$, and a sample count tradeoff. We jointly optimize modality, resolution, power, and sample count, and demonstrate the resulting gains via simulations. The journal version extends to sensing-communication-computation co-design under practical link impairments, imperfect sensing, and over-the-air aggregation.


\balance 

\bibliographystyle{IEEEtran}
\bibliography{IEEEabrv,ref}

@STRING{IEEE_J_VT         = "{IEEE} Trans. Veh. Technol."}

@STRING{IEEE_J_SP         = "{IEEE} Trans. Signal Processing"}

@STRING{IEEE_J_STSP       = "{IEEE} J. Sel. Topics Signal Process."}

@STRING{IEEE_J_WCOM       = "{IEEE} Trans. Wireless Commun."}

@STRING{IEEE_J_JSAIT       = "{IEEE} J. Select. Areas Inform. Theory"}

@STRING{IEEE_J_MC          = "{IEEE} Trans. Mobile Comput."}

@ARTICLE{Saad-6G-2020,
  author={Saad, Walid and Bennis, Mehdi and Chen, Mingzhe},
  journal={IEEE Network}, 
  title={A Vision of {6G} Wireless Systems: Applications, Trends, Technologies, and Open Research Problems}, 
  year={2020},
  volume={34},
  number={3},
  pages={134-142},
  }

@ARTICLE{Giordani-6G-ComMag2020,
  author={Giordani, Marco and Polese, Michele and Mezzavilla, Marco and Rangan, Sundeep and Zorzi, Michele},
  journal={IEEE Communications Magazine}, 
  title={Toward {6G} Networks: Use Cases and Technologies}, 
  year={2020},
  volume={58},
  number={3},
  pages={55-61},
  }

@ARTICLE{Brinton-6G-ComMag2025,
  author={Brinton, Christopher G. and Chiang, Mung and Kim, Kwang Taik and Love, David J. and Beesley, Michael and Repeta, Morris and Roese, John and Beming, Per and Ekudden, Erik and Li, Clara and Wu, Geng and Batra, Nishant and Ghosh, Amitava and Ziegler, Volker and Ji, Tingfang and Prakash, Rajat and Smee, John},
  journal={IEEE Communications Magazine}, 
  title={Key Focus Areas and Enabling Technologies for {6G}}, 
  year={2025},
  volume={63},
  number={3},
  pages={84-91},
  }

@ARTICLE{Andrews-6G,
  author={Andrews, Jeffrey G. and Humphreys, Todd E. and Ji, Tingfang},
  journal={IEEE BITS the Information Theory Magazine}, 
  title={{6G} Takes Shape}, 
  year={2024},
  volume={4},
  number={1},
  pages={2-24},
  }

@inproceedings{McMahan2017FedAvg,
  author    = {McMahan, Brendan and Moore, Eider and Ramage, Daniel and Hampson, Seth and Aguera y Arcas, Blaise},
  title     = {Communication-Efficient Learning of Deep Networks from Decentralized Data},
  booktitle = {Proc. International Conference on Artificial Intelligence and Statistics (AISTATS)},
  year      = {2017},
  pages     = {1273--1282}
}

@article{Li2020FLSurvey,
  author  = {Li, Tian and Sahu, Anit Kumar and Talwalkar, Ameet and Smith, Virginia},
  title   = {Federated Learning: Challenges, Methods, and Future Directions},
  journal = {IEEE Signal Processing Magazine},
  year    = {2020},
  volume  = {37},
  number  = {3},
  pages   = {50--60}
}

@article{Kairouz2021AdvancesFL,
  author  = {Kairouz, Peter and McMahan, H. Brendan and Avent, Brendan and others},
  title   = {Advances and Open Problems in Federated Learning},
  journal = {Foundations and Trends in Machine Learning},
  year    = {2021},
  volume  = {14},
  number  = {1--2},
  pages   = {1--210}
}

@article{Amiri2020ADSGD,
  author  = {Amiri, Mohammad Mohammadi and Gunduz, Deniz},
  title   = {Machine Learning at the Wireless Edge: Distributed Stochastic Gradient Descent Over-the-Air},
  journal = IEEE_J_SP,
  year    = {2020},
  volume  = {68},
  pages   = {2155--2169}
}

@article{Yang2020AirCompFL,
  author  = {Yang, Kai and Jiang, Tao and Shi, Yuanming and Ding, Zhi},
  title   = {Federated Learning via Over-the-Air Computation},
  journal = IEEE_J_WCOM,
  year    = {2020},
  volume  = {19},
  number  = {3},
  pages   = {2022--2035}
}

@article{Sery2021OTAHetero,
  author  = {Sery, Tomer and Shlezinger, Nir and Cohen, Kobi and Eldar, Yonina C.},
  title   = {Over-the-Air Federated Learning from Heterogeneous Data},
  journal = IEEE_J_SP,
  year    = {2021},
  volume  = {69},
  pages   = {3796--3811}
}

@article{Liang2025ISCC,
  author  = {Liang, Yipeng and Chen, Qimei and Zhu, Guangxu and Eldar, Yonina C. and Cui, Shuguang},
  title   = {Communication-and-Energy Efficient Over-the-Air Federated Learning},
  journal = IEEE_J_WCOM,
  year={2025},
  volume={24},
  number={1},
  pages={767-782},
}

@article{Fu2025JSCC,
  author  = {Fu, Yang and Qin, Peng and Tang, Guoming and Zhao, Xiongwen},
  title   = {Joint Design of Sensing, Communication, and Computation for Multi-{AAV}-Enabled Over-the-Air Federated Learning},
  journal = IEEE_J_VT,
  year={2025},
  volume={74},
  number={9},
  pages={13909-13924}
}

@article{Cao2026VFEEL,
  author={Cao, Xiaowen and Wen, Dingzhu and Bi, Suzhi and Cui, Yuanhao and Zhu, Guangxu and Hu, Han and Eldar, Yonina C.},
  journal=IEEE_J_MC, 
  title={Joint Sensing, Communication, and Computation for Vertical Federated Edge Learning in Edge Perception Networks}, 
  year={2026},
  volume={},
  number={},
  pages={1-14}
}

@ARTICLE{Karbalayghareh-JSAIT26-dynamicFL,
  author={Karbalayghareh, Mehdi and Love, David J. and Brinton, Christopher G.},
  journal=IEEE_J_JSAIT, 
  title={Coherence-Aware Over-the-Air Distributed Learning Under Heterogeneous Link Impairments}, 
  year={2026},
  volume={7},
  number={},
  pages={46-61},
  }

@ARTICLE{Liu-STSP23_taskOriented,
  author={Liu, Peixi and Zhu, Guangxu and Wang, Shuai and Jiang, Wei and Luo, Wu and Poor, H. Vincent and Cui, Shuguang},
  journal=IEEE_J_STSP, 
  title={Toward Ambient Intelligence: Federated Edge Learning With Task-Oriented Sensing, Computation, and Communication Integration}, 
  year={2023},
  volume={17},
  number={1},
  pages={158-172},
 }

@ARTICLE{Wen-TWC2026-ISCC,
  author={Wen, Dingzhu and Xie, Sijing and Cao, Xiaowen and Cui, Yuanhao and Xu, Jie and Shi, Yuanming and Cui, Shuguang},
  journal=IEEE_J_WCOM, 
  title={Integrated Sensing, Communication, and Computation for Over-the-Air Federated Edge Learning}, 
  year={2026},
  volume={25},
  number={},
  pages={2748-2762},
  }

@ARTICLE{mmFedMC,
  author={Yuan, Liangqi and Han, Dong-Jun and Wang, Su and Upadhyay, Devesh and Brinton, Christopher G.},
  journal=IEEE_J_MC, 
  title={Communication-Efficient Multimodal Federated Learning: Joint Modality and Client Selection}, 
  year={2026},
  volume={},
  number={},
  pages={1-18},
}

@article{RELIEF,
  author = {B. Wu and Z. Ding and J. Huang},
  title = {{RELIEF}: Turning Missing Modalities into Training Acceleration for Federated Learning on Heterogeneous {IoT} Edge},
  journal = {arXiv preprint arXiv:2604.04243},
  year = {2026}
}

@article{Asaad-Tabassum-2026,
  author = {S. Asaad and H. Tabassum and P. Wang},
  title = {Multi-objective Optimization for Over-the-Air Federated Edge Learning-enabled Collaborative Integrated Sensing and Communications},
  journal = {arXiv preprint arXiv:2603.15783},
  year = {2026}
}

@INPROCEEDINGS{EM-PFL-ISAC-2025,
  author={Ni, Zhou and Chintareddy, Sravan Reddy and Guan, Peiyuan and Hashemi, Morteza},
  booktitle={2026 IEEE 23rd Consumer Communications \& Networking Conference (CCNC)}, 
  title={Personalized Federated Learning-Driven Beamforming Optimization for Integrated Sensing and Communication Systems}, 
  year={2026},
  volume={},
  number={},
  pages={1-6},
  }

@inproceedings{Li2020FedProx,
  author = {T. Li and A. K. Sahu and M. Zaheer and M. Sanjabi and A. Talwalkar and V. Smith},
  title = {Federated Optimization in Heterogeneous Networks},
  booktitle = {Proc. Machine Learning and Systems (MLSys)},
  year = {2020}
}

@inproceedings{Karimireddy-SCAFFOLD,
  author = {S. P. Karimireddy and S. Kale and M. Mohri and S. J. Reddi and S. U. Stich and A. T. Suresh},
  title = {{SCAFFOLD}: Stochastic Controlled Averaging for Federated Learning},
  booktitle = {Proc. Int. Conf. Machine Learning (ICML)},
  year = {2020}
}

\end{document}